\documentclass[a4paper,onecolumn,11pt]{quantumarticle}
\pdfoutput=1
\usepackage[utf8]{inputenc}
\usepackage[english]{babel}
\usepackage[T1]{fontenc}
\usepackage{amsmath}
\usepackage{hyperref}
\usepackage{amsthm}
\usepackage{tikz}
\usepackage[misc]{ifsym}
\usepackage[numbers,sort&compress]{natbib}
\newtheorem{theorem}{Theorem}
\theoremstyle{remark}\newtheorem{lemma}{\hskip 1em Lemma}
\begin{document}

\title{Robust self-testing of multipartite GHZ-state measurements in quantum networks}

\author{Qing Zhou}
\author{Xin-Yu Xu}
\author{Shuai Zhao}
\author{Yi-Zheng Zhen}
\author{Li Li\footnotemark[1]$^{\ ,}$}
\author{Nai-Le Liu\footnotemark[2]$^{\ ,}$}
\author{Kai Chen\footnotemark[3]$^{\ ,}$}
\affiliation{Hefei National Laboratory for Physical Sciences at Microscale and Department of Modern Physics, University of Science and Technology of China, Hefei, Anhui 230026, People's Republic of China}
\affiliation{CAS Center for Excellence and Synergetic Innovation Center in Quantum Information and Quantum Physics, University of Science and Technology of China, Hefei, Anhui 230026, People's Republic of China}

\renewcommand{\thefootnote}{\fnsymbol{footnote}}
\footnotetext[1]{\href{mailto:eidos@ustc.edu.cn}{eidos@ustc.edu.cn}}
\footnotetext[2]{\href{mailto:eidos@ustc.edu.cn}{nlliu@ustc.edu.cn}}
\footnotetext[3]{\href{mailto:eidos@ustc.edu.cn}{kaichen@ustc.edu.cn}}

\maketitle
\begin{abstract}
Self-testing is a device-independent examination of quantum devices based on correlations of observed statistics. Motivated by elegant progresses on self-testing strategies for measurements [Phys.\ Rev.\ Lett.\ $\textbf{121}$, 250507 (2018)] and for states [New J.\ Phys.\ $\textbf{20}$,\ $083041$ (2018)], we develop a general self-testing procedure for multipartite generalized GHZ-state measurements. The key step is self-testing all measurement eigenstates for a general $N-$qubit multipartite GHZ-state measurement. Following our procedure, one only needs to perform local measurements on $N-2$ chosen parties of the $N$-partite eigenstate and maintain to achieve the maximal violation of tilted Clauser-Horne-Shimony-Holt (CHSH) Bell inequality for remaining two parties. Moreover, this approach is physically operational from an experimental point of view. It turns out that the existing result for three-qubit GHZ-state measurement is recovered as a special case. Meanwhile, we develop the self-testing method to be robust against certain white noise.
\end{abstract}

\section{Introduction}\label{sec:introduction}
The rapid development of quantum communication in recent years creates an exigent requirement 
for devising certification methods 
to guarantee correctness of quantum information tasks. To rule out any potential attacks by malicious third party, such certification methods must be device-independent. As the first device-independent tool, the Bell nonlocality has been extensively studied in recent decades \cite{brunner2014bell}.\ It has brought great breakthroughs in quantum physics. Recently, as the strongest form of device-independent certification, self-testing has been developed, which is also based on Bell nonlocality. Such certification method can characterize the target objects (quantum states, measurements) fully, only up to local isometries, in a device-independent manner.

Self-testing, acting as a device-independent certification method, has attracted lots of attention since the pioneer works of Mayers and Yao \cite{mayers2003self}. It can be used to certify entangled pure states and measurements \cite{yang2013robust,natarajan2016robust,mckague2012robust,bamps2015sum,vsupic2016self,vsupic2019device,goswami2018one,sarkar2019self,bharti2019robust,li2020self,coladangelo2018generalization,coopmans2019robust,kaniewski2017self,tavakoli2018self,wu2016device,yang2014robust,miklin2020semi,jebarathinam2019maximal,miklin2021universal,bharti2021graph}. Up to now, a wide range of entangled quantum states are proved to be self-testable, such as the elegant results for all pure bipartite entangled states \cite{coladangelo2017all}, three-qubit W states \cite{wu2014robust}, and graph states \cite{mckague2011self}. It has also been shown that all pure multipartite GHZ states and Dicke states can be self-tested \cite{vsupic2018self}. Recently, the self-testing method for quantum channels has also been developed \cite{sekatski2018certifying}. Moreover, there have been many applications about self-testing, such as quantum key distribution \cite{mckague2010generalized}, randomness expansion \cite{miller2016robust}, detection for entanglement \cite{bowles2018device}, certification of genuinely entangled subspaces \cite{baccari2020device,makuta2021self}, coarse-grained self-testing of a many-body singlet \cite{PhysRevLett.127.240401}, as well as verification of quantum computations \cite{gheorghiu2017rigidity,mckague2013interactive}. 

In this work, we will focus on self-testing entangled measurements in quantum networks. Self-testing entangled quantum measurements is of great potential to develop practical quantum networks, which has been preliminarily studied \cite{renou2018self,bancal2018noise}. For a star-network as shown in Fig.$~$1, where $N$ observers share entangled states with central node, respectively. A self-testable entangled measurement can guarantee the success of quantum information tasks, which are based on distributing entangled states between remote parties in such a network. Meanwhile, the entanglement between observers and central node can also be certified. In Ref.$~$\cite{renou2018self}, the authors presented a self-testing method for the Bell-state measurement (BSM) and three-qubit GHZ-state measurement (GSM). Furthermore, a more robust self-testing scheme for BSM has also been proposed in Ref.$~$\cite{bancal2018noise}. However, there have not been a detailed characterization for self-testing multipartite ($N>3$) entangled measurements directly.

\begin{figure}[ht]
\centering
\includegraphics[scale=0.3]{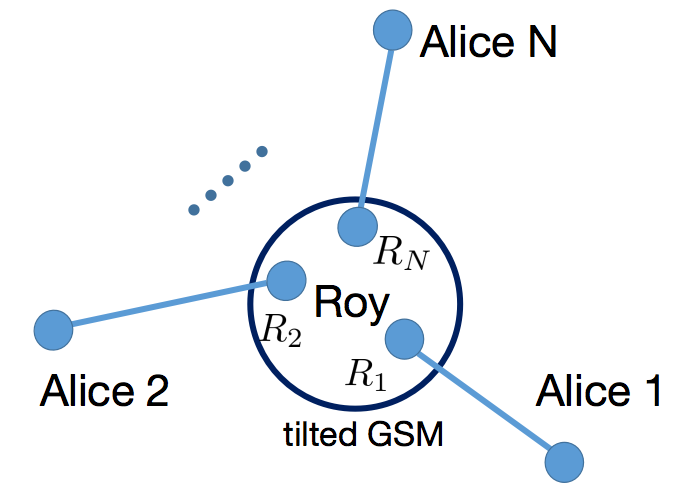}
\caption{Roy shares a Bell state with each of the other $N$ observers (Alice $1$, Alice $2$,$\dots$, Alice $N$). If Roy performs multipartite tilted GSM, then the state shared by Alice $1$, Alice $2$,$\dots$, Alice $N$ will be projected into $|GHZ^{r}_{\theta}\rangle$. Conversely, if Alice $1$, Alice $2$,$\dots$, Alice $N$ observe that they are projected into $|GHZ^{r}_{\theta}\rangle$ for all $r$, the measurement performed by Roy is equivalent to the tilted GSM.}
\label{fig1}
\end{figure}
By generalizing the idea of Ref.$~$\cite{renou2018self}, we present herewith a self-testing method for $N$-qubit tilted GSM, whose eigenstates are partially entangled $N$-qubit GHZ states (tilted GHZ states). That is to say, one of the eigenstates of $N$-qubit tilted GSM can be written as $|GHZ^{0}_{\theta}\rangle=\cos\theta|0\rangle^{\otimes N}+\sin\theta|1\rangle^{\otimes N}, \ \theta\in(0,\pi/4].$ In the $N+1$-partite star-network shown in Fig.$~$1, after performing measurement in central node (Roy), the remaining particles shared with Alice $1$, Alice $2$,$\dots$, and Alice $N$ will be projected into the eigenstates of Roy's measurement. If the particular correlations of remaining particles are observed by local measurements on Alice $1$, Alice $2$,$\dots$, and Alice $N$, the measurement performed by Roy is equivalent to an $N$-qubit tilted GSM, up to a local isometry. A local isometry is a linear local operation on quantum states that preserves inner products \cite{vsupic2020self}. Thus, the first step to self-test tilted GSM is self-testing all its measurement eigenstates, and the problem of self-testing tilted GSM can be converted to the problem of self-testing states. Motivated by the method of self-testing multipartite entangled states in Ref.$~$\cite{vsupic2018self}, we have developed further a general method for self-testing multipartite tilted GSM in a star-network, and the method is operational from an experimental point of view. We also show that one can self-test more entangled measurements by our developed method straightforwardly.

The paper is organized as following. In Sec.$~$\ref{sec:pre}, we provide a preparation review of tilted CHSH scenario which constitute important ingredient of our self-testing method. In Sec.$~$\ref{sec:multi}, self-testing method of multipartite tilted GSM is presented. In Sec.$~$\ref{sec:robust}, a noise-robust self-testing scheme of three-qubit GSM is presented, with the help of semidefinite program (SDP) method. Finally, we conclude our results and make a discussion on potential future works in Sec.$~$\ref{sec:conclude}. 

\section{Preliminaries}\label{sec:pre}
To self-test tilted multipartite GSM, the tilted CHSH inequality is necessary \cite{acin2012randomness}. Let us consider a task: Alice and Bob share a two-qubit state and they want to know whether the shared state is partially entangled or not. They perform local measurements (dichotomic observables) respectively. The tilted CHSH inequality is given by
\begin{equation}
\alpha\langle A_0\rangle+\langle A_{0}B_{0}\rangle+\langle A_{0}B_{1}\rangle+\langle A_{1}B_{0}\rangle-\langle A_{1}B_{1}\rangle\leq 2+\alpha,
\end{equation}
where the maximal value of violation is $\sqrt{8+2\alpha^{2}}, \alpha\in [0,2)$, $A_{i}$ and $B_{i}$ being observables with outcomes $\{-1,+1\}$ measured locally by Alice and Bob. Here, we omit the notation ``$\otimes$'' between systems $A$ and $B$ and write $A_{0}\otimes I$ as $A_{0}$ for short. After performing local measurements, if Alice and Bob obtain the maximal violation of tilted CHSH inequality, the state shared by them is a certain partially entangled two-qubit state (tilted Bell state). For detailed case, the four tilted Bell states $|Bell^{b}_{\theta}\rangle, b=0,1,2,3$ are given by
\begin{equation}
\nonumber
\begin{split}
&|Bell^{0}_{\theta}\rangle=c_\theta|00\rangle+s_\theta|11\rangle, \ |Bell^{1}_{\theta}\rangle=s_\theta|00\rangle-c_\theta|11\rangle,\\&|Bell^{2}_{\theta}\rangle=c_\theta|01\rangle+s_\theta|10\rangle, \ |Bell^{3}_{\theta}\rangle=s_\theta|01\rangle-c_\theta|10\rangle,
\end{split}
\end{equation}
where $c_{\theta}=\cos\theta,s_{\theta}=\sin\theta, \theta\in (0,\frac{\pi}{4}]$. Let $\mu$ satisfy $\tan\mu=\sin2\theta$ and $\sigma_{Z},\sigma_{X}$ be Pauli matrices. 
If one fixes the measurement settings of Alice and Bob as $A_{0}=\sigma_{Z},\  A_{1}=\sigma_{X},\ B_{0}=\cos\mu\sigma_{Z}+\sin\mu\sigma_{X},\ B_{1}=\cos\mu\sigma_{Z}-\sin\mu\sigma_{X}$, the output statistics obtained by these measurements will maximally violate some tilted CHSH inequalities. The maximal violation is $CHSH_{b}^{\alpha}=\langle Bell^{b}_{\theta}| W_{b}^{\alpha}|Bell^{b}_{\theta}\rangle=\sqrt{8+2\alpha^{2}}$ with $\alpha=2\cos 2\theta/\sqrt{1+\sin^{2}2\theta}$, where
\begin{equation}
\begin{split}
&W_{0}^{\alpha}=  \alpha A_{0}+A_{0}B_{0}+A_{0}B_{1}+A_{1}B_{0}-A_{1}B_{1},\notag\\& 
W_{1}^{\alpha}=-\alpha A_{0}+A_{0}B_{0}+A_{0}B_{1}-A_{1}B_{0}+A_{1}B_{1},\notag\\
&W_{2}^{\alpha}=-W_{1}^{\alpha},
W_{3}^{\alpha}=-W_{0}^{\alpha}.\notag
\end{split}
\end{equation}
Here the $W_{b}^{\alpha}$ is Bell operator acting on the Hilbert space $\mathcal{H}_{A}\otimes\mathcal{H}_{B}$ of Alice and Bob. It is easy to show that the eigenvalue $\sqrt{8+2\alpha^{2}}$ of the Bell operator $W_{b}^{\alpha}$ is nondegenerate with associated eigenvector $|Bell^{b}_{\theta}\rangle$. Hence, if the maximal violation of $CHSH_{b}^{\alpha}$ is $\sqrt{8+2\alpha^{2}}$, the shared state will be $|Bell^{b}_{\theta}\rangle$. One can discriminate the four tilted Bell states by the maximally violations of four tilted Bell inequality with fix measurement settings. Furthermore, other tilted Bell states that are local-unitary (constructed by $\sigma_{Z},\sigma_{X}$) equivalent to the above four tilted Bell states can also be discriminated. For example, the state $|\Phi\rangle=c_{\theta}|00\rangle-s_{\theta}|11\rangle=\sigma_{Z_{A}}|Bell_{\theta}^{0}\rangle$. It can maximally violate tilted CHSH inequality with $CHSH_{\theta}=\langle\sigma_{Z_{A}}W_{0}^{\alpha}\sigma_{Z_{A}}^{\dag}\rangle=\alpha \langle A_{0}\rangle+\langle A_{0}B_{0}+A_{0}B_{1}-A_{1}B_{0}+A_{1}B_{1}\rangle$ and fixed measurements given above.

In the entanglement swapping scenario \cite{zukowski1993event} shown in Fig.$~$2, let Charlie perform tilted BSM whose measurement eigenstates are tilted Bell states with outcomes $b$. Then, the remaining state will be projected into one of the four tilted Bell states $|Bell^{b}_{\theta}\rangle$ conditioned on the outcomes $b$. Conversely, if one finds that Alice and Bob share tilted Bell states $|Bell^{b}_{\theta}\rangle$ for $b\in\{0,1,2,3\}$, the performed measurement of Charlie is a tilted BSM. Motivated by this idea, we will develop a procedure for preforming self-testing of tilted multipartite GSM.
\begin{figure}[ht]
\centering
\includegraphics[scale=0.3]{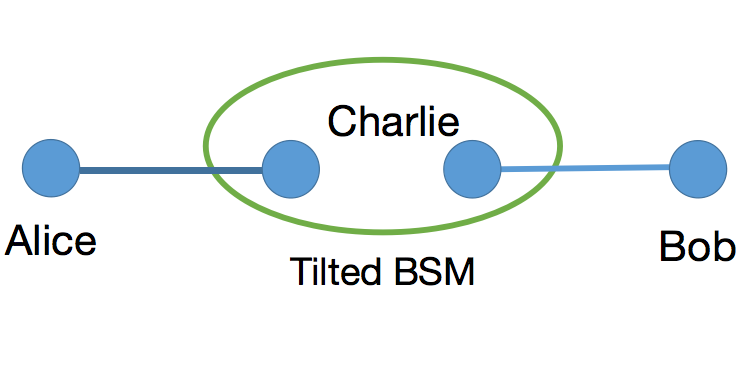}
\caption{An entanglement swapping scenario: Charlie shares a maximally entangled two-qubit state with each of the other two observers (Alice and Bob). If Charlie performs tilted BSM and obtains outcome b, then Alice and Bob will be projected into $|Bell^{b}_{\theta}\rangle$, i.e.,  Alice and Bob can observe the maximal violation of the specific tilted CHSH inequality with $CHSH_b^{\alpha}=\sqrt{8+2\alpha^{2}}$. }
\label{fig2}
\end{figure}

\section{Self-test tilted multipartite GSM}\label{sec:multi}
As shown in Ref.$~$\cite{choi1975completely}, any completely positive and trace preserving (CPTP) map can be implemented, by tracing out degrees of freedom that does not involve effective information after applying a local isometry. Therefore, one can adopt the approach presented in \cite{kaniewski2017self,renou2018self,vsupic2020self} to present the definition for self-testing multipartite measurements via simulation: denote an ideal $d$-outcome measurement for Roy acting on $\mathcal{H}_{R_{1}^{\prime}}\otimes\mathcal{H}_{R_{2}^{\prime}}\otimes\dots\otimes\mathcal{H}_{R_{N}^{\prime}}$ as $\mathcal{P}^{\prime}=\{P_{R_{1}^{\prime}R_{2}^{\prime}\dots R_{N}^{\prime}}^{\prime r}\}_{r=1}^{d}$, and a real measurement acting on $\mathcal{H}_{R_{1}}\otimes\mathcal{H}_{R_{2}}\otimes\dots\otimes\mathcal{H}_{R_{N}}$ as $\mathcal{P}=\{P_{R_{1}R_{2}\dots R_{N}}^{r}\}_{r=1}^{d}$. If there exist completely positive and unital maps $\Lambda_{R_{j}}:\ \mathcal{L}(\mathcal{H}_{R_{j}})\to \mathcal{L}(\mathcal{H}_{A_{j}^{\prime}})$, for $j\in\{1,2,\dots,N\}$, such that 
\begin{equation}
\Lambda_{R_{1}}\otimes\Lambda_{R_{2}}\otimes\dots\otimes\Lambda_{R_{N}}(P_{R_{1}R_{2}\dots R_{N}}^{r})=P_{R_{1}^{\prime}R_{2}^{\prime}\dots R_{N}^{\prime}}^{\prime r},
\label{eq:cptp}
\end{equation}
for all $r$, we say $\mathcal{P}$ is capable of simulating $\mathcal{P}^{\prime}$. In the above definition, we adopt the assumption that the different physical sources are independent in a quantum network. The construction of a quantum network as shown in Fig.$~$1 guarantees the well-defined $N$-partition for Roy's measurement device, i.e., $\mathcal{H}_{R}=\mathcal{H}_{R_{1}}\otimes\mathcal{H}_{R_{2}}\otimes\dots\otimes\mathcal{H}_{R_{N}}$.

The idea of our self-testing method relies on the task of entanglement swapping as shown in Fig.$~$1. There are $N$ initially uncorrelated parties Alice $1$ ($A_{1}$), Alice $2$ ($A_{2}$),$\dots$, Alice $N$ ($A_{N}$). They are independently entangled with an additional party, Roy. Specifically, the $A_{i}$ and Roy share a Bell state $|Bell^{0}_{\pi/4}\rangle_{A_{i}R_{i}}\in \mathcal{H}_{A_{i}}\otimes\mathcal{H}_{R_{i}}, i\in\{1,2,\dots,N\}$. To distribute entanglement between $A_{1},A_{2},\dots,A_{N}$ in such a quantum network, Roy performs the tilted GSM and obtains outcomes $r\in\{k_{1}k_{2}\dots k_{N}\}$ with $k_{1},k_{2},\cdots,k_{N}\in\{0,1\}$. For simplicity, we denote the outcomes as $r\in\{0,1,\dots,2^{N}-1\}$. Then, the states shared by $A_{1},A_{2},\dots,A_{N}$ are projected to one of the $2^{N}$ tilted GHZ states $|GHZ^{r}_{\theta}\rangle$ based on the outcome $r$.\ The tilted GHZ states are measurement eigenstates of tilted GSM given by 
\begin{equation}\nonumber
\begin{split}
|GHZ^{r}_{\theta}\rangle=(-1)^{\sum\limits_{i=1}^{N}k_{i}}&\cos\theta|k_{1}k_{2}\dots k_{N}\rangle+\sin\theta|\bar{k}_{1}\bar{k}_{2}\dots\bar{k}_{N}\rangle,
\end{split}
\end{equation}
where $\bar{k}_{i}=1-k_{i}$ and $k_{i}\in\{0,1\},\ i\in\{1,\dots,N\}$.\ The tilted GSM can be denoted as $GSM_{\theta}=\{GHZ^{r}_{\theta}\}_{r=0}^{2^{N}-1}$, with $GHZ^{r}_{\theta}=|GHZ^{r}_{\theta}\rangle\langle GHZ^{r}_{\theta}|$.\ If $A_{1},A_{2},\dots,A_{N}$ obtain the outcomes $r$ of Roy, they can apply a special local unitary operation on their qubits, so that they share a special tilted GHZ state. With the above operations, we have implemented distribution of entanglement between $N$ remote parties.

The self-testing procedure is similar to the task of entanglement swapping, without assumptions on the dimensions, initial states and operators. From now on, let us adopt labels on the same letter to make a distinction between two Hilbert spaces, e.g., $Q$ and $Q^{\prime}$. Specially, the $Q^{\prime}$ is in a two-dimensional Hilbert space and the dimension of $Q$ is unknown. Let us start with presenting a self-testing method for $N$-partite tilted GHZ states given in Ref.$~$\cite{vsupic2018self}.
\begin{lemma}
(Please Refer to Ref.$~$\cite{vsupic2018self}). Suppose an $N$-partite state $|\psi\rangle$, and a pair of binary observables $A_{0,i}$, $A_{1,i}$ for the $i$-th party, for $i=1,...,N$. For an observable $D$, let $P_{D}^{a}=[I+(-1)^{a}D]/2, a\in\{0,1\}$. Let $\mu$ satisfy $\tan\mu=\sin 2\theta$, $Z_{i}=A_{0,i}$, and $X_{i}=A_{1,i}$, for $i=1,...,N-1$. Then, let $Z_{N}^{*}$ be  $(A_{0,N}+A_{1,N})/(2\cos\mu)$ with zero eigenvalues replaced by 1 and $X_{N}^{*}$ be $(A_{0,N}-A_{1,N})/(2\sin\mu)$ with zero eigenvalues replaced by 1. Define $Z_{N}=Z_{N}^{*}|Z_{N}^{*}|^{-1}$ and $X_{N}=X_{N}^{*}|X_{N}^{*}|^{-1}$. If the following relations are satisfied:
\begin{gather}
\langle\psi|P_{A_{0,i}}^{0}|\psi\rangle=\langle\psi|P_{A_{0,i}}^{0}P_{A_{0,j}}^{0}|\psi\rangle=c_{\theta}^{2}, \ \ \forall i,j\in \{1,...,N-1\},\notag\\ \langle\psi|\prod_{i=1}^{N-2} P_{A_{1,i}}^{a_{i}}|\psi\rangle=\frac{1}{2^{N-2}},  \ \ \forall a_{i}\in \{0,1\},\notag\\ \langle\psi|(\prod_{i=1}^{N-2} P_{A_{1,i}}^{a_{i}})(\alpha A_{0,N-1}\otimes I_{N}+A_{0,N-1}A_{0,N}+A_{0,N-1}A_{1,N}+\notag\\ (-1)^{\sum_{i=1}^{N-2}a_{i}}(A_{1,N-1}A_{0,N}-A_{1,N-1}A_{1,N}))|\psi\rangle =  \frac{\sqrt{8+2\alpha^{2}}}{2^{N-2}},\notag\\ \forall a_{i}\in \{0,1\}, \notag
\end{gather}
where $\alpha=2\cos 2\theta/\sqrt{1+\sin^{2}2\theta}$ and $c_{\theta}=\cos\theta,\theta\in(0,\pi/4]$, there exists a local isometry $\Phi$ such that
\begin{equation}\nonumber
\Phi(|\psi\rangle)=|junk\rangle |GHZ_{\theta}^{0}\rangle,
\end{equation}
for some junk state $|junk\rangle$. 
Hence, these relations for correlations self-test the state $|GHZ_{\theta}^{0}\rangle=\cos\theta|0\rangle^{\otimes N}+\sin\theta|1\rangle^{\otimes N}$. 
\end{lemma}
The junk state in the Lemma 1 can be any state and can be removed by tracing out the $A_{1}A_{2}\dots A_{N}$ space. It should be noted that this self-testing method is also suitable for a general $\rho$ \cite{vsupic2018self}. Without loss of generality, let the $N$-partite state be a pure state. Here, the $Z_{N}$ and $X_{N}$ act on $|\psi\rangle$ in the same way as $(A_{0,N}+A_{1,N})/(2\cos\mu)$ and $(A_{0,N}-A_{1,N})/(2\sin\mu)$, respectively \cite{vsupic2020self}. For details, the ideal measurements achieving these correlations in the Lemma 1 are: $A_{0,i}^{\prime}=\sigma_{Z}$, $A_{1,i}^{\prime}=\sigma_{X}$, for $i=1,\dots ,N-1$, and $A_{0,N}^{\prime}=\cos \mu\sigma_{Z}+\sin\mu\sigma_{X}$, $A_{1,N}^{\prime}=\cos \mu\sigma_{Z}-\sin\mu\sigma_{X}$. 

From the Lemma 1, all partially entangled $N$-partite GHZ states can be self-tested by checking whether the projected state of the remaining two parties ($A_{N-1}$ and $A_{N}$) maximally violates the tilted CHSH inequality. The remaining two parties are the parties after performing local measurements on the other $N-2$ parties. Moreover, for different $r\in\{0,1,\dots,2^{N}-1\}$, the $|GHZ_{\theta}^{r}\rangle$ can all be self-tested by correlations in the Lemma 1 with different measurement settings up to local isometries. In other words, one can obtain a local isometry, such that $\Phi^{r}(|\psi^{r}\rangle)=|junk\rangle|GHZ_{\theta}^{r}\rangle,$ for each $r$. As the isometry can always be constructed by local operations which does not depend on $r$, one can always construct a single isometry, such that $\Phi(|\psi^{r}\rangle)=|junk\rangle|GHZ_{\theta}^{r}\rangle.$ The detailed description will be shown in the next lemma. 

Now, let us firstly introduce some notations. For an observable $O^{\prime }$ acting on Hilbert space $\mathcal{H}^{\prime}=\otimes_{i=1}^{N}\mathcal{H}_{A_{i}^{\prime}}$, let $\widetilde{O^{\prime}}^{r}=U^{\prime r\dag}O^{\prime }U^{\prime r}$, where $U^{\prime r}=\otimes_{i=1}^{N}{U^{\prime r}_{A_{i}^{\prime}}}$ acting on $\mathcal{H}^{\prime}$. The $\mathcal{H}_{A_{i}^{\prime}},i\in\{1,2,\dots,N\}$, are two dimensional Hilbert spaces. The unitary operator $U^{\prime r}$ satisfies the equation $U^{\prime r}|GHZ^{r}_{\theta}\rangle=|GHZ^{0}_{\theta}\rangle$ and is constructed by the product of identity matrix $I^{\prime}$, and Pauli matrices $X^{\prime},Z^{\prime}$. Then, one can define $\widetilde{O}^{r}=U^{ r\dag}OU^{ r}$ by replacing the $I^{\prime},X^{\prime},Z^{\prime}$ in $\widetilde{O^{\prime}}^{r}$ with $I,X,Z$. By the above special unitary transformation, one can obtain following Lemma 2.
\begin{lemma}
Let $|\psi\rangle$ be an $N$-partite state, and let $A_{0,i}$, $A_{1,i}$ be a pair of binary observables for the $i$-th party, for $i=1,...,N$. Suppose that, for all $r\in\{0,1,\dots,2^{N}-1\}$, the following relations are satisfied:
\begin{gather}
\langle\psi|\widetilde{P_{A_{0,i}}^{0}}^{r}|\psi\rangle=\langle\psi|\widetilde{P_{A_{0,i}}^{0}}^{r}\widetilde{P_{A_{0,j}}^{0}}^{r}|\psi\rangle=c^{2}_{\theta}, \ \forall i,j\in \{1,...,N-1\},\\ \langle\psi|\prod_{i=1}^{N-2} \widetilde{P_{A_{1,i}}^{a_{i}}}^{r}|\psi\rangle=\frac{1}{2^{N-2}},  \ \ \forall a_{i}\in \{0,1\},\\ \langle\psi|(\prod_{i=1}^{N-2} \widetilde{P_{A_{1,i}}^{a_{i}}}^{r})\widetilde{W_{\bar{a}}^{\alpha}}^{r}|\psi\rangle =  \frac{\sqrt{8+2\alpha^{2}}}{2^{N-2}}, \forall a_{i}\in \{0,1\} \label{eq:lemma1}
\end{gather}
where $\bar{a}\equiv a_{1}\dots a_{N-2}$ and 
\begin{equation}\nonumber
\begin{split}
\widetilde{W_{\bar{a}}^{\alpha}}^{0}=W_{\bar{a}}^{\alpha}=&\alpha A_{0,N-1}\otimes I_{N}+A_{0,N-1}A_{0,N}+A_{0,N-1}A_{1,N}+ \\&(-1)^{\Sigma_{i=1}^{N-2}a_{i}}(A_{1,N-1}A_{0,N}-A_{1,N-1}A_{1,N}).
\end{split}
\end{equation}
The detailed forms for $\widetilde{P_{A_{0,i}}^{0}}^{r},\widetilde{P_{A_{1,i}}^{a_{i}}}^{r}$ are easy to calculate and the details for $\widetilde{W_{\bar{a}}^{\alpha}}^{r}$ as an example are provided in the Appendix B. The measurements here are the same as shown in the Lemma 1. Then, there exists a single local isometry such that $\Phi(|\psi^{r}\rangle)=|junk\rangle|GHZ_{\theta}^{r}\rangle,$ for all $r$.
\end{lemma}
\begin{proof}
For $r=0$, the correlations in the Lemma 2 are same as the Lemma 1. Hence these correlations self-test state $|GHZ^{0}_{\theta}\rangle$. Denote $|\psi\rangle$ in the self-testing procedure as $|\psi^{0}\rangle$. From the Lemma 1, there exists a local isometry $\Phi$ such that $\Phi(|\psi^{0}\rangle)=|junk\rangle|GHZ^{0}_{\theta}\rangle$. Meanwhile, $X_{f}^{2}=Z_{f}^{2}=I$ and $X_f,Z_f$ anti-commute over the support of the state $|\psi^{0}\rangle$, for all $f\in\{A_{1},A_{2},\dots,A_{N}\}$ \cite{vsupic2018self}. Then, one can construct this isometry by ancillary qubits $|0\rangle^{\otimes N}$ and swap gates $\{S_{X_{f},Z_{f}}\}$ as
\begin{equation} 
\Phi(|\psi^{0}\rangle)=|junk\rangle|GHZ^{0}_{\theta}\rangle=(\otimes_{i=1}^{N}S_{X_{A_{i}},Z_{A_{i}}})|0\rangle^{\otimes N}|\psi^{0}\rangle.
\end{equation}
The detailed form of a swap gate is shown in Fig.$~$3. From the Lemma 1 in Ref.$~$\cite{renou2018self}, one knows that $S_{X_{f},Z_{f}}\cdot X\cdot |0\rangle |\xi_{f}\rangle=X^{\prime}\cdot S_{X_{f},Z_{f}}\cdot |0\rangle |\xi_{f}\rangle$ and $S_{X_{f},Z_{f}}\cdot Z\cdot |0\rangle |\xi_{f}\rangle=  Z^{\prime}\cdot S_{X_{f},Z_{f}}\cdot |0\rangle |\xi_{f}\rangle$. Let $S_{A_{1}A_{2}\dots A_{N}}=(\otimes_{i=1}^{N}S_{X_{A_{i}},Z_{A_{i}}})$. As the $U^{r}$ is constructed by $I,X,Z$, one has 
\begin{equation}\nonumber
\begin{split}
\Phi(U^{r\dag}|\psi^{0}\rangle)&=S_{A_{1}A_{2}\dots A_{N}}|0\rangle^{\otimes N} U^{r\dag}|\psi^{0}\rangle\\&=U^{\prime r\dag}S_{A_{1}A_{2}\dots A_{N}}  |0\rangle^{\otimes N}|\psi^{0}\rangle \\&=U^{\prime r\dag}\Phi(|\psi^{0}\rangle)=|junk\rangle\otimes U^{\prime r\dag}|GHZ^{0}_{\theta}\rangle\\&=|junk\rangle\otimes|GHZ^{r}_{\theta}\rangle.
\end{split}
\end{equation}
Here $U^{r\dag}|\psi^{0}\rangle=|\psi\rangle$. One has $\Phi(|\psi\rangle)=|junk\rangle|GHZ^{r}_{\theta}\rangle$. Therefore, the relations for correlations in the Lemma 2 self-test state $|GHZ^{r}_{\theta}\rangle$. The $|\psi\rangle$ can be denoted as $|\psi^{r}\rangle$. Thus, one has $\Phi(|\psi^{r}\rangle)=|junk\rangle|GHZ^{r}_{\theta}\rangle$

\begin{figure}[ht]
\centering
\includegraphics[scale=0.3]{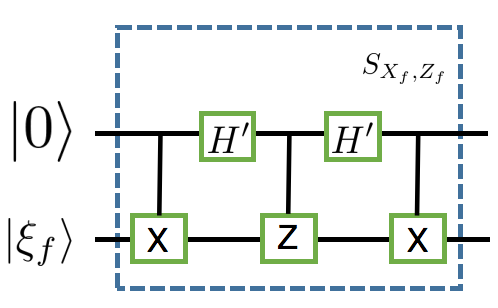}
\caption{Swap gate is constructed by unitary $X$, $Z$ and $H^{\prime}$, where $H^{\prime}$ is the Hadamard gate, and $X$, $Z$ are anti-commute over the support of the state $|\xi_{f}\rangle\in\mathcal{H}$. The $|0\rangle$ is in the qubit Hilbert space $\mathcal{H}^{\prime}$. }
\label{fig3}
\end{figure}
\end{proof}
From the Lemma 2, self-testing method with fixed measurements can be used to distinguish special entangled pure states. Here, let $\{|GHZ_{\theta}^{r}\rangle\}_{r=0}^{2^{N}-1}$ be reference states and $|GHZ_{\theta}^{0}\rangle$ be a standard reference state. For example, there is a set of states $\{|\psi^{s}\rangle\}_{s=0}^{2^{N}-1}$ shared by $A_{1},A_{2},\dots,A_{N}$. If one shared state $|\psi^{s_{1}}\rangle$ satisfies the correlations in the Lemma 2 with $r=0$, one can specify the shared state $|\psi^{s_{1}}\rangle$ as state $|\psi^{0}\rangle$ according to standard reference state $|GHZ_{\theta}^{0}\rangle$. Then, for another shared state $|\psi^{s_{2}}\rangle$ with $s_{2}\in\{0,1,2,\dots s_{1}-1,s_{1}+1,\dots,2^{N}-1\}$ , if it satisfies correlations in the Lemma 2 for one $r$ with $r\in\{1,2,\dots,2^{N}-1\}$, e.g., $r=3$, then, one resets the $s_{2}$ as $s_{2}=3$. In other words, the state $|\psi^{s_{2}}\rangle$ can be rewritten as $|\psi^{3}\rangle$ and these correlations have self-tested the $|GHZ_{\theta}^{3}\rangle$. Therefore, the states $|\psi^{s_{1}}\rangle$ and $|\psi^{s_{2}}\rangle$ are actually different. Now, the main result of the paper is following. 
\begin{theorem}
Let $A_{1},A_{2},\dots,A_{N}$ share respectively a pair of quantum state with Roy as $\tau_{A_{1}R_{1}A_{2}R_{2}\dots A_{N}R_{N}}=\tau_{A_{1}R_{1}}\otimes\tau_{A_{2}R_{2}}\otimes\cdots\otimes\tau_{A_{N}R_{N}}$ and let $\mathcal{R}=\{R_{R_{1}R_{2}\cdots R_{N}}^{r}\}_{r=0}^{2^{N}-1}$ be a $2^{N}$-outcome measurement acting on $\mathcal{H}_{R_{1}}\otimes\mathcal{H}_{R_{2}}\otimes\cdots\otimes\mathcal{H}_{R_{N}}$. For the $A_{1},A_{2},\dots,A_{N}$, if there exist measurements such that the observed correlations conditioned on outcome $r$ of Roy's measurement satisfy the relations in the Lemma 2, then there exist completely positive and unital maps $\Lambda_{R_{i}}:\ \mathcal{L}(\mathcal{H}_{R_{i}})\to \mathcal{L}(\mathcal{H}_{A_{i}^{\prime}})$, $i\in\{1,2,\dots,N\}$, for $dim(\mathcal{H}_{A_{i}^{\prime}})=2$ such that
\begin{equation}
\Lambda_{R_{1}}\otimes\Lambda_{R_{2}}\otimes\cdots\otimes\Lambda_{R_{N}}(R_{R_{1}R_{2}\cdots R_{N}}^{r})=GHZ_{\theta}^{r}
\label{eq:theoremm}
\end{equation}
for $r\in\{0,1,2,\dots,2^{N}-1\}$.
\end{theorem}
The detailed proof is shown in Appendix A. Here, we present a brief description. Let the $\tau_{A_{1}A_{2}\dots A_{N}}^{r}=|\psi^{r}\rangle\langle \psi^{r}|$ acting on $\otimes_{i=1}^{N}\mathcal{H}_{i}$ be the state shared by $A_{1},A_{2},\dots,A_{N}$ conditioned on outcome $r$. From the Lemma 2, there exists a single isometry such that $\Phi(|\psi^{r}\rangle)=|junk\rangle|GHZ_{\theta}^{r}\rangle$. By tracing out the subsystems $\mathcal{H}_{1},\dots,\mathcal{H}_{N}$, one can construct a single pair of swap channels $\Gamma_{A_{i}}:\mathcal{L}(\mathcal{H}_{A_{i}})\to\mathcal{L}(\mathcal{H}_{A_{i}^{\prime}}),i\in\{1,2,\dots,N\}$, such that
\begin{equation}\nonumber
(\otimes_{i=1}^{N}\Gamma_{A_{i}})(\tau_{A_{1}A_{2}\dots A_{N}}^{r})=|GHZ_{\theta}^{r}\rangle\langle GHZ_{\theta}^{r}|,
\end{equation}
for all $r$. With the help of Choi-Jamio\l kowski map \cite{renou2018self}, one can construct completely positive and unital maps which are associated with above swap channels, such that
\begin{equation}\nonumber
\begin{split}
(\otimes_{i=1}^{N}\Lambda_{R_{i}})(R_{R_{1}R_{2}\cdots R_{N}}^{r})=(\otimes_{i=1}^{N}\Gamma_{A_{i}})(\tau_{A_{1}A_{2}\dots A_{N}}^{r})=GHZ_{\theta}^{r}.
\end{split}
\end{equation}
The $2^{N}$ equations given by Eq.$~$(\ref{eq:theoremm}) imply that a real measurement $\mathcal{R}=\{R_{R_{1}\dots R_{N}}^{r}\}_{r=0}^{2^{N}-1}$ is capable of simulating ideal tilted GSM, $\{GHZ_{\theta}^{r}\}_{r=0}^{2^{N}-1}$, i.e., the Theorem 1 self-tests the tilted GSM. 
The method presents a unified form of the theorem for multipartite case without resorting to different Bell inequalities. Furthermore, one can also self-test multipartite GSM, if $\alpha=0, \theta=\pi/4$. Moreover, if $N=3$, one can recover the case of three-qubit GSM \cite{renou2018self}.

Remarkably, for any self-testing method of tilted GHZ-states, if the ideal measurements in the self-testing procedure are constructed by Pauli matrices, it can be adopted to self-test tilted GSM. Such a property can be a rule to construct the self-testing method for tilted GSM.

\section{Robust self-testing of the GSM}\label{sec:robust}
The ideal self-testing method is an excellent tool to device-independently certify quantum information tasks. However, due to the imperfection of quantum devices, the accurate correlations in the above theorem may not be satisfied. Hence, a robust version of self-testing is necessary from an experimental point of view. For convenience, we will study here a robust self-testing scheme of three-qubit GSM, where $N=3,\alpha=0, \theta=\pi/4$. The method for studying robustness of other cases is similar.

Before presenting the robustness of GSM, let us firstly study the robust self-testing of the GHZ state with semi-definite programs (SDP) method. One can rewrite $A_{1},A_{2},A_{3}$ as $A,B,C$ and let $A_{i,1}=A_{i},A_{i,2}=B_{i},A_{i,3}=C_{i},i\in\{0,1\}$. 
Let the state shared by $A,\ B$ and $C$ with outcome $r=0$ be $\tau_{ABC}^{0}=|\psi^{0}\rangle\langle\psi^{0}|$. In a general way, one can adopt the fidelity $F=\langle GHZ|\sigma_{A^{\prime}B^{\prime}C^{\prime}}^{0}|GHZ\rangle$ to capture the distance of the unknown state to the target state \cite{li2018self}, 
where $|GHZ\rangle=\frac{|000\rangle+|111\rangle}{\sqrt{2}}$ and $\sigma_{A^{\prime}B^{\prime}C^{\prime}}^{0}=\Gamma_{A}\otimes\Gamma_{B}\otimes\Gamma_{C}(\tau_{ABC}^{0})$. The maps $\Gamma_{f},f\in\{A,B,C\}$ are defined from Fig.$~$3 as 
$\Gamma_{f}(|\xi\rangle_{f} \langle\xi|)=Tr_{\mathcal{H}_{f}}(S_{X_{f},Z_{f}}|0\rangle\langle 0|\otimes |\xi\rangle_{f}\langle \xi|S_{X_{f},Z_{f}}^{\dag})$ with $f\in\{A,B,C\}$. Here, the assumption that $X,Z$ are anti-commutative in the definition of $\Gamma$ has been removed. The state $\sigma_{A^{\prime}B^{\prime}C^{\prime}}^{0}$ can be written as
\begin{equation}
\sigma_{A^{\prime}B^{\prime}C^{\prime}}^{0}=Tr_{ABC}(S_{ABC}|000\rangle_{A^{\prime}B^{\prime}C^{\prime}}\langle 000|\otimes\tau_{ABC}^{0}S_{ABC}^{\dag}).
\end{equation}
From the definition of fidelity, one has
\begin{equation}
\begin{split}
F&=\langle GHZ|\sigma_{A^{\prime}B^{\prime}C^{\prime}}^{0}|GHZ\rangle\\&=\frac{1}{128}Tr_{ABC}\{8(1+Z_{A})(1+Z_{B})(1+Z_{C})\tau_{ABC}^{0}\\& \quad+8(1-Z_{A})(1-Z_{B})(1-Z_{C})\tau_{ABC}^{0}\\&\quad+(\Pi_{f\in\{A,B,C\}}(1+Z_{f})X_{f}(1-Z_{f}))\tau_{ABC}^{0}\\&\quad+(\Pi_{f\in\{A,B,C\}}(1-Z_{f})X_{f}(1+Z_{f}))\tau_{ABC}^{0}\},
\end{split}
\end{equation}
where the fidelity can be expressed as a linear function of the expectation values. Suppose the channel suffers with white noise (weight $\epsilon$), one can transform the problem of robustness into the problem that finding a lower bound on the fidelity. It can be solved by SDP \cite{wu2014robust,li2018self,bancal2015physical,vandenberghe1996semidefinite}:
\begin{equation}
\begin{split}
&min \ \ F=\langle GHZ|\sigma_{A^{\prime}B^{\prime}C^{\prime}}^{0}|GHZ\rangle,\\&s.t. \ \ M\geq 0,\\&\qquad \langle\psi|P_{A_{0}}^{0}|\psi\rangle=\langle\psi|P_{B_{0}}^{0}|\psi\rangle=\frac{1}{2},\\&\qquad\langle\psi|P_{A_{0}}^{0}P_{B_{0}}^{0}|\psi\rangle=\frac{1-\epsilon}{2}+\frac{\epsilon}{4},\\&\qquad   \langle\psi|P_{A_{1}}^{a}|\psi\rangle=\frac{1}{2},\ \ for \ a\in\{0,1\}, \\&\qquad   \langle\psi|P_{A_{1}}^{a}(\alpha B_0+B_{0}C_{0}+B_{0}C_{1}\\&\qquad+(-1)^{a}(B_{1}C_{0}-B_{1}C_{1}))|\psi\rangle=\sqrt{2}(1-\epsilon), 
\label{eq:epsilon}
\end{split}
\end{equation}
where $M$ is a moment matrix defined by $M_{ij}=Tr(\tau_{ABC}^{0}D_{i}^{\dag}D_{j})$ with set $\{D_{1}=I,D_{2}=Z_{A},D_{3}=X_{A}\dots\}$ \cite{navascues2008convergent}. 
For an ideal case, the fidelity is 1 when error $\epsilon=0$. For other $\epsilon$ up to 0.1225, the relations between minimal fidelity and error are shown in Fig.$~$4. Thus, the Fig.$~$4 gives a lower bound of fidelity for different $\epsilon$. Without loss of generality, one can define the relation between minimal fidelity and $\epsilon$ as a function $G(\epsilon^{0})$, which will be used to study the robustness of GSM. Here, the $\epsilon$ has been rewritten as $\epsilon^{0}$.

\begin{figure}[ht]
\centering
\includegraphics[scale=0.3]{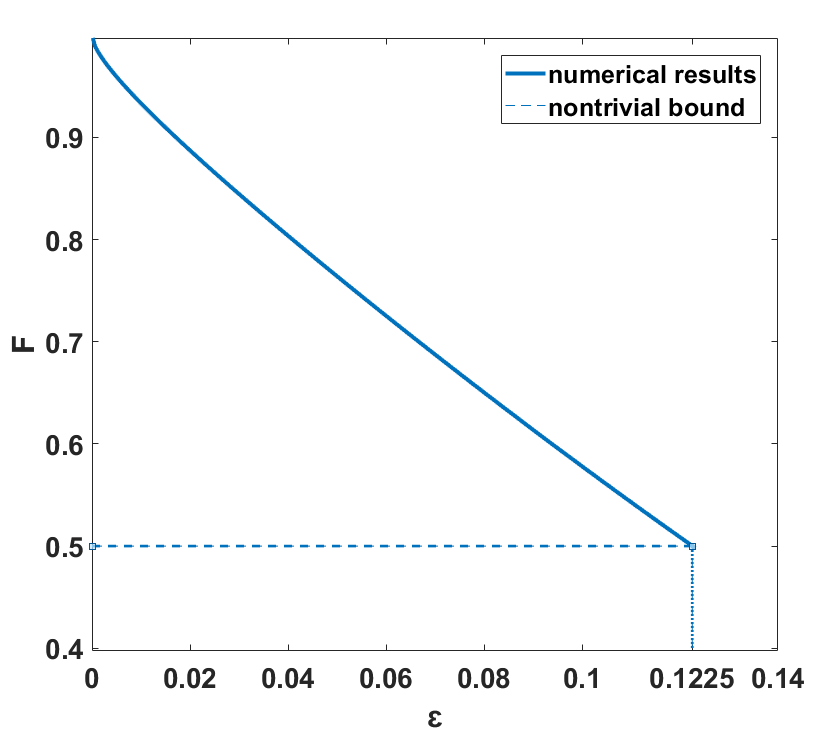}
\caption{The lower bound on the fidelity $F$ between GHZ state and unknown state $\sigma_{A^{\prime}B^{\prime}C^{\prime}}^{0}$ for different white noise $\epsilon$. When the fidelity is above the nontrivial bound of $0.5$ (i.e., $\epsilon \leq 12.25\%$), the unknown state is close to a GHZ state.}
\label{fig4}
\end{figure}

For defining quality of real measurement $\mathcal{R}$ as a simulation of ideal GSM $\mathcal{P}$, where $\mathcal{R}=\{R_{R_{1}R_{2}R_{3}}^{r}\}_{0}^{7}$ and $\mathcal{P}=\{GHZ^{r}\}_{0}^{7}$, we directly extend the definition in Ref.$~$\cite{renou2018self} to three parties as
\begin{equation}
\begin{split}
&\mathcal{Q}(\mathcal{R},\mathcal{P})\\&=\frac{1}{8}\times\mathop{max}\limits_{\Lambda_{R_{1}}\Lambda_{R_{2}}\Lambda_{R_{3}}}\sum\limits_{r=0}^{7}\langle (\Lambda_{R_{1}}\otimes\Lambda_{R_{2}}\otimes\Lambda_{R_{3}})(R_{R_{1}R_{2}R_{3}}^{r}),GHZ^{r}\rangle.
\label{eq:quality}
\end{split}
\end{equation}
Here, we omit the subscript of $GHZ_{\frac{\pi}{4}}^{r}$ as $GHZ^{r}$ and $\Lambda_{R_{1}},\ \Lambda_{R_{2}},\ \Lambda_{R_{3}}$ are unital CPTP maps with $\Lambda_{R_{1}}:\ \mathcal{L}(\mathcal{H}_{R_{1}})\to \mathcal{L}(\mathcal{H}_{A^{\prime}}),\ \Lambda_{R_{2}}:\ \mathcal{L}(\mathcal{H}_{R_{2}})\to \mathcal{L}(\mathcal{H}_{B^{\prime}}),\ \Lambda_{R_{3}}:\ \mathcal{L}(\mathcal{H}_{R_{3}})\to \mathcal{L}(\mathcal{H}_{C^{\prime}})$. The symbol $\langle\ ,\ \rangle$ for two matrices $L_1$ and $L_2$ implies 
\begin{equation}\nonumber
\langle L_1,L_2\rangle=Tr(L_{1}L_{2}^{\dag}).
\end{equation}
Now, the robust version of self-testing method is presented as following.
\begin{theorem}
Let $A,\ B,\ C$ share a pair of quantum state with Roy respectively as $\tau_{AR_{1}BR_{2}CR_{3}}=\tau_{AR_{1}}\otimes\tau_{BR_{2}}\otimes\tau_{CR_{3}}$ and let $\mathcal{R}=\{R_{R_{1}R_{2}R_{3}}^{r}\}_{r=0}^{7}$ be a 8-outcome measurement acting on $\mathcal{H}_{R_{1}}\otimes\mathcal{H}_{R_{2}}\otimes\mathcal{H}_{R_{3}}$. Let $p_{r}$ be the probability of Roy observing the outcome $r$. Define the function $G(\epsilon^{r})$ as the lower bound on the fidelity between $\Gamma_{A}\otimes\Gamma_{B}\otimes\Gamma_{C}(\tau_{ABC}^{r})$ and $GHZ^{r}$ under noise $\epsilon^{r}$. 
For $A,\ B$ and $C$, suppose there exist measurements, such that the observed correlations conditioned on outcomes $r$ satisfy the relations in the Lemma 2 with error $\epsilon^{r}$ and $G(\epsilon^{r})> 0.5$. Define $q=\Sigma_{r}p_{r}G(\epsilon^{r})$, then one has
\begin{equation}
\begin{split}
\mathcal{Q}(\mathcal{R},&\mathcal{P})\geq\frac{1}{2(1+2\sqrt{q(1-q)})^{2}}\cdot\\&\mathop{min}\limits_{u\in[0,2\sqrt{q(1-q)}]}(\frac{2q-1}{\sqrt{(1-u^{2})}}+\frac{1}{(1+u)}).
\end{split}
\end{equation} 
\end{theorem}
The detailed proof is given in Appendix C. One can always let every $\epsilon^{r}$ be $max\{\epsilon^{r}\}_{r=0}^{7}$ and denote it as $\epsilon$. Then, one has $q=G(\epsilon)$, which can be obtained by numerical method of SDP problem. The relation between quality of unknown real measurement and the noise $\epsilon=max\{\epsilon^{r}\}_{r=0}^{7}$ is shown in Fig.$~$5. Thus, we have shown the robust self-testing scheme of the GSM with the noise tolerance up to\ $~0.28\%$. From the definition of quality $\mathcal{Q}(\mathcal{R},\mathcal{P})$ (\ref{eq:quality}), it should go through all possible unital CPTP maps $\Lambda_{R_{1}},\ \Lambda_{R_{2}},\ \Lambda_{R_{3}}$ and then choose the maximal value. However, our result is currently based on only one choice of these maps. Hence, if one optimizes this question and finds the maximum result, a better robustness can be expected. With the help of SDP method, one can straightforwardly obtain the robust version of our self-testing method for multipartite tilted GSM, similar to the robust self-testing method done for three-qubit GSM here. 
\begin{figure}[ht]
\centering
\includegraphics[scale=0.3]{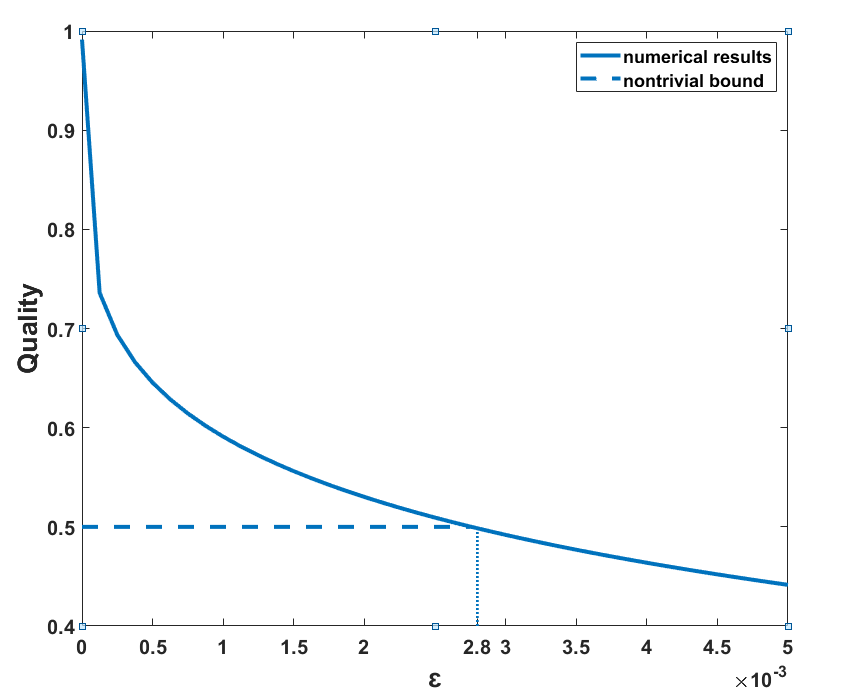}
\caption{The lower bound on the quality of the unknown real measurement is numerically estimated as a function about the weight of white noise $\epsilon$. When the weight of white noise $\epsilon\leq 0.28\%$ (i.e., the quality is above the nontrivial bound of 0.5), the presented procedure guarantees the unknown measurement is close to a three-qubit GHZ-state measurement.}
\label{fig5}
\end{figure}

\section{Conclusion}\label{sec:conclude}
In quantum network, it is extremely vital to certify multipartite entangled measurements. Here, we have presented the first self-testing method for the important class of general GHZ-state measurements. The procedure is operational for arbitrary number of parties from experimental point of views, and does not resort to the common method of verifying $N$-partite Bell inequalities. Meanwhile, the approach can recover the case of three-qubit GHZ-state measurement directly. In addition, we have provided robustness of the self-testing procedure with the help of semi-definite program. The noise tolerance is up to 0.28$\%$ when certifying a three-qubit GHZ-state measurement.

For future works, it is interesting to develop more robust method to open the possibility to estimate the robustness of arbitrary multipartite entangled measurements, and enable experiments about self-testing quantum networks. It is expected that our approach can also be extended to high dimensional case, as the self-testing method done for high dimensional entangled states \cite{vsupic2018self}.

\section{Acknowledgements}
We sincerely thank Xinhui Li for insightful discussions about the technology of semi-definite programs. This work has been supported by the National Natural Science Foundation of China (Grants No. 62031024, 11874346), the National Key R$\& $D Program of China (2019YFA0308700) and the Anhui Initiative in Quantum Information Technologies (AHY060200).

\appendix

\section*{Appendix A: Proof of the Theorem 1}\label{sec:appen1}
\setcounter{equation}{0}
\renewcommand\theequation{A\arabic{equation}}
As shown in the Theorem 1, if the observed correlations conditioned on outcome of Roy's measurement satisfy the relations in the Lemma 2, the measurement performed by Roy is a tilted GHZ-state measurement. Now, let us present the detailed proof of it.
\begin{proof}
Let $p_{r}$ be the probability of Roy observing the outcome $r$, and $\tau^{r}_{A_{1}A_{2}\dots A_{N}}=|\psi^{r}\rangle_{A_{1}A_{2}\dots A_{N}}\langle\psi^{r}|$ be the state shared between $A_{1}\dots A_{N}$ conditioned on outcome $r\in\{0,\dots,2^{N}-1\}$, i.e., $p_{r}\tau_{A_{1}\dots A_{N}}^{r}=Tr_{R_{1}R_{2}\dots R_{N}}[(I_{A_{1}A_{2}\dots A_{N}}\otimes R_{R_{1}R_{2}\dots R_{N}}^{r})(\otimes_{i=1}^{N}\tau_{A_{i}R_{i}})]$. One can always choose $p_{r}=\frac{1}{2^{N}}$. By the definition of swap gate in Fig.$~$3, one can construct swap channels as 
\begin{equation}\nonumber
\Gamma_{f}(|\xi\rangle_{f} \langle\xi|)=Tr_{\mathcal{H}_{f}}(S_{X_{f},Z_{f}}|0\rangle\langle 0|\otimes |\xi\rangle_{f}\langle \xi|S_{X_{f},Z_{f}}^{\dag}),
\end{equation}
where $f\in\{A_{1},A_{2},\dots A_{N}\}$. Define
\begin{equation}
\begin{split}
&\sigma_{A_{i}^{\prime}R_{i}}\equiv\Gamma_{A_{i}}(\tau_{A_{i}R_{i}}), i\in\{1,2,\dots,N\},\\
&\sigma_{A_{1}^{\prime}A_{2}^{\prime}\dots A_{N}^{\prime}}^{r}\equiv(\otimes_{i=1}^{N}\Gamma_{A_{i}})(\tau_{A_{1}A_{2}\dots A_{N}}^{r})\\&\quad\quad \ \ \ =(\frac{1}{p_{r}})Tr_{R_{1}R_{2}\dots R_{N}}(R_{R_{1}R_{2}\dots R_{N}}^{r}(\otimes_{i=1}^{N}\sigma_{A_{i}^{\prime}R_{i}}))\\&\quad\quad \ \ \ =(2^{N})Tr_{R_{1}R_{2}\dots R_{N}}(R_{R_{1}R_{2}\dots R_{N}}^{r}(\otimes_{i=1}^{N}\sigma_{A_{i}^{\prime}R_{i}})).
\end{split}
\label{eq:abc}
\end{equation}
Then, one has
\begin{equation}
\begin{split}
&\quad(\Gamma_{A_{1}}\otimes\Gamma_{A_{2}}\otimes\dots\otimes\Gamma_{A_{N}})(\tau_{A_{1}A_{2}\dots A_{N}}^{r})\\&=Tr_{A_{1}A_{2}\dots A_{N}}(S_{A_{1}A_{2}\dots A_{N}}|0\rangle^{\otimes N}_{A_{1}^{\prime}\dots A_{N}^{\prime}}\langle 0|^{\otimes N}\otimes\tau_{A_{1}A_{2}\dots A_{N}}^{r}S_{A_{1}A_{2}\dots A_{N}}^{\dag})\\&=Tr_{A_{1}A_{2}\dots A_{N}}(S_{A_{1}A_{2}\dots A_{N}}|0\rangle^{\otimes N}_{A_{1}^{\prime}\dots A_{N}^{\prime}}|\psi\rangle^{r}\langle 0|^{\otimes N}\langle\psi^{r}|S_{A_{1}A_{2}\dots A_{N}}^{\dag})\\&=Tr_{A_{1}A_{2}\dots A_{N}}(|junk\rangle_{A_{1}A_{2}\dots A_{N}}\langle junk|\otimes|GHZ_{\theta}^{r}\rangle\langle GHZ_{\theta}^{r}|)\\&=|GHZ_{\theta}^{r}\rangle\langle GHZ_{\theta}^{r}|.
\label{eq:gammaqubit}
\end{split}
\end{equation} 
The third equality is from the Lemma 2. From the definition of the state $\sigma_{A_{1}^{\prime}A_{2}^{\prime}\dots A_{N}^{\prime}}^{r}$, one has 
\begin{equation}\nonumber
\sigma_{A_{1}^{\prime}A_{2}^{\prime}\dots A_{N}^{\prime}}^{r}=GHZ_{\theta}^{r},
\end{equation}
for all $r\in\{0,1,\dots,2^{N}-1\}$. Let us firstly present the definition of Choi-Jamio\l kowski map \cite{renou2018self}. If $\rho_{AB}$ acts on $\mathcal{H}_{A} \otimes\mathcal{H}_{B}$, the Choi-Jamio\l kowski map ($\Lambda_{B}:\mathcal{H}_B\to\mathcal{H}_A$) associated to it is defined by $\Lambda_{B}(\sigma_B)=Tr_{B}[(I_A\otimes\sigma_{B}^{T})\rho_{AB}]$ for $\forall \sigma_{B}$. Here, $\rho_{AB}$ is the Choi state and can be unnormalized. Now, let $\Lambda_{R_{i}}:\ \mathcal{L}(\mathcal{H}_{R_{i}})\to \mathcal{L}(\mathcal{H}_{A_{i}^{\prime}})$, be respectively the Choi-Jamio\l kowski maps associated to the operators $2\sigma_{A_{i}^{\prime}R_{i}},i\in\{1,2,\dots,N\}$. By decomposing the operator $R_{R_{1}R_{2}\dots R_{N}}^{r}$ as $R_{R_{1}R_{2}\dots R_{N}}^{r}=\Sigma_{l}\bigotimes_{k}\omega_{k,l}^{r}$, where $\omega_{k,l}^{r}$ is the operator of $\mathcal{H}_{R_{k}}$, one has $\Lambda_{R_{1}}\otimes\Lambda_{R_{2}}\otimes\dots\otimes\Lambda_{R_{N}}(R_{R_{1}R_{2}\dots R_{N}}^{r})=(2^{N})Tr_{R_{1}R_{2}\dots R_{N}}(R_{R_{1}R_{2}\dots R_{N}}^{r}(\otimes_{i=1}^{N}\sigma_{A_{i}^{\prime}R_{i}}))=\sigma_{A_{1}^{\prime}A_{2}^{\prime}\dots A_{N}^{\prime}}^{r}=GHZ_{\theta}^{r}$. Moreover, we will prove that these Choi maps $\Lambda_{R_{i}},i\in\{1,2,\dots,N\}$, are unital maps. Let us first consider $\Lambda_{R_{1}}$, the other cases being similar. By the definition of Choi-Jamio\l kowski map, one has 
\begin{equation}\nonumber
\begin{split}
\Lambda_{R_{1}}(I_{R_{1}})&=Tr_{R_{1}}(2\sigma_{A_{1}^{\prime}R_{1}})\\&=2Tr_{R_{1}R_{2}\dots R_{N}A_{2}^{\prime}\dots A_{N}^{\prime}}(\otimes_{i=1}^{N}\sigma_{A_{i}^{\prime}R_{i}})\\&=2\Sigma_{r=0}^{2^{N}-1}Tr_{R_{1}R_{2}\dots R_{N}A_{2}^{\prime}\dots A_{N}^{\prime}}(R_{R_{1}R_{2}R_{3}}^{r}(\otimes_{i=1}^{N}\sigma_{A_{i}^{\prime}R_{i}}))\\&=\frac{1}{2^{N-1}}\Sigma_{r=0}^{2^{N}-1}Tr_{A_{2}^{\prime}\dots A_{N}^{\prime}}(\sigma_{A_{1}^{\prime}A_{2}^{\prime}\dots A_{N}^{\prime}}^{r})=I_{A_{1}^{\prime}},
\end{split}
\end{equation} 
where we have used the fact that $\Sigma_{r=0}^{2^{N}-1}R_{R_{1}R_{2}\dots R_{N}}^{r}=I$ and $\sigma_{A_{1}^{\prime}A_{2}^{\prime}\dots A_{N}^{\prime}}^{r}=GHZ_{\theta}^{r}$.
\end{proof}
Therefore, we have proven that the joint measurement performed by central node Roy is actually a tilted GHZ-state measurement under the conditions in the Lemma 2.

\section*{Appendix B: The detailed form of $\widetilde{W_{\bar{a}}^{\alpha}}^{r}$ in the Lemma 2}\label{sec:appen2}
\setcounter{equation}{0}
\renewcommand\theequation{B\arabic{equation}}
In the Lemma 2, a new form of self-testing statement has been presented. The notation ``${\widetilde{\quad}}^{r}$'' in the ${\widetilde{O}}^{r}$ means that local unitary transformations are performed on the observable $O$. Here, we will provide the details of $\widetilde{W_{\bar{a}}^{\alpha}}^{r}$, where $\bar{a}=a_{1}\cdots a_{N-2}$. For convenience, let $N=3$. Rewrite $A_{1},\ A_{2},\ A_{3}$ as $A,\ B,\ C$ and let $A_{i,1}=A_{i},\ A_{i,2}=B_{i},\ A_{i,3}=C_{i},\ i\in\{0,1\}$. Now, the $\bar{a}$ is $a_{1}$, and rewritten as $a\in\{0,1\}$. The $W_0^{\alpha}$ and $W_1^{\alpha}$ can been obtained from the Lemma 2 as $W_0^{\alpha}=\alpha B_{0}\otimes I+ B_{0}C_{0}+B_{0}C_{1}+B_{1}C_{0}-B_{1}C_{1}, W_1^{\alpha}=\alpha B_{0}\otimes I+  B_{0}C_{0}+B_{0}C_{1}-B_{1}C_{0}+B_{1}C_{1}$. Firstly, by adding superscript in the formula of $W_0^{\alpha}$ and $W_1^{\alpha}$, one has $W_{0}^{\prime\alpha}=\alpha B_{0}^{\prime}+W_{0}^{\prime}, W_{1}^{\prime\alpha}=\alpha B_{0}^{\prime}+W_{1}^{\prime}$, where $W_{0}^{\prime}=B_{0}^{\prime}C_{0}^{\prime}+B_{0}^{\prime}C_{1}^{\prime}+B_{1}^{\prime}C_{0}^{\prime}-B_{1}^{\prime}C_{1}^{\prime},\ W_{1}^{\prime}=B_{0}^{\prime}C_{0}^{\prime}+B_{0}^{\prime}C_{1}^{\prime}-B_{1}^{\prime}C_{0}^{\prime}+B_{1}^{\prime}C_{1}^{\prime}$. From the Lemma 1, one knows that $B_{0}^{\prime}=Z^{\prime},B_{1}^{\prime}=X^{\prime}$ and $C_{0}^{\prime}=\cos\mu Z^{\prime}+\sin\mu X^{\prime},C_{1}^{\prime}=\cos\mu Z^{\prime}-\sin\mu X^{\prime}$ with $\tan\mu=\sin 2\theta$. The local unitary transformation performed on $W_{a}^{\prime\alpha}$ is $U^{\prime r}=U_{A^{\prime}}^{\prime r}\otimes U^{\prime r}_{B^{\prime}C^{\prime}}$. As the $U^{\prime r}$ is the local unitary transformation between tilted GHZ states, one can always choose $U^{\prime r}_{B^{\prime}C^{\prime}}\in\{X^{\prime}\otimes X^{\prime},I^{\prime}\otimes X^{\prime},X^{\prime}\otimes I^{\prime},I^{\prime}\otimes I^{\prime}\}$. For $r=7$ case, one has $X^{\prime}Z^{\prime}\otimes X^{\prime}\otimes X^{\prime}|GHZ_{\theta}^{7}\rangle=|GHZ_{\theta}^{0}\rangle$, where $|GHZ_{\theta}^{7}\rangle=\sin\theta|000\rangle-\cos\theta|111\rangle$ and $|GHZ_{\theta}^{0}\rangle=\cos\theta|000\rangle+\sin\theta|111\rangle$. Here, the $U^{\prime 7}=X^{\prime}Z^{\prime}\otimes X^{\prime}\otimes X^{\prime}$. Thus, $\widetilde{W_{0}^{\prime\alpha}}^{7}=U^{\prime 7\dag}W_{0}^{\prime\alpha}U^{\prime 7}=-\alpha B_{0}^{\prime}+W_{0}^{\prime}$ and $\widetilde{W_{1}^{\prime\alpha}}^{7}=U^{\prime 7\dag}W_{1}^{\prime\alpha}U^{\prime 7}=-\alpha B_{0}^{\prime}+W_{1}^{\prime}$. After calculating $\widetilde{W_{a}^{\prime\alpha}}^{r}$ for all $r\in\{0,1,\dots,7\}$ and $a\in\{0,1\}$, the detailed formulas of $\widetilde{W_{a}^{\prime\alpha}}^{r}$ can be obtained. By replacing the symbols $I^{\prime},B_{i}^{\prime},C_{i}^{\prime},i\in\{0,1\}$ in $\widetilde{W_{a}^{\prime\alpha}}^{r}$ with $I,B_{i},C_{i},i\in\{0,1\}$, one can obtain the detailed form of $\widetilde{W_{a}^{\alpha}}^{r}$.

In short, the $\widetilde{W_{\bar{a}}^{\alpha}}^{r}$ is acquired by deleting the superscript prime of $\widetilde{W_{\bar{a}}^{\prime\alpha}}^{r}$. The $\widetilde{W_{\bar{a}}^{\prime\alpha}}^{r}$ is obtained by performing local unitary transformations on $W_{\bar{a}}^{\prime\alpha}$. The local unitary transformation depends on the transformation between states $|GHZ_{\theta}^{r}\rangle$ and $|GHZ_{\theta}^{0}\rangle$. Therefore, one can easily write the detailed form of $\widetilde{W_{\bar{a}}^{\alpha}}^{r}$ in the Lemma 2.

\section*{Appendix C: Proof of the Theorem 2}\label{appen3}
\setcounter{equation}{0}
\renewcommand\theequation{C\arabic{equation}}
In this section, we give a proof of the Theorem 2 that shows the robust self-testing of three-qubit GHZ-state measurement. If the observed correlations can not perfectly satisfy the conditions in the Lemma 2, one can not adopt the ideal self-testing method presented in the Theorem 1 directly. We should bound the quality of the unknown measurement under the certain white noise, i.e., study how close the unknown measurement performed by Roy to ideal three-qubit GHZ-state measurement. 
Before presenting proof of the Theorem 2, we firstly generalize the result of semi-definite program in main text as following lemma.
\begin{lemma}
Let $A_{0},\ A_{1},\ B_{0},\ B_{1},\ C_{0},\ C_{1}$, be the pairs of observables for the three parties. If the correlations in the Lemma 2 with error $\epsilon^{r}$ ($\theta=\pi/4,\alpha=0$) satisfy the following relations: 
\begin{gather}
\langle\psi|\widetilde{P_{A_{0}}^{0}}^{r}|\psi\rangle=\langle\psi|\widetilde{P_{B_{0}}^{0}}^{r}|\psi\rangle=\frac{1}{2},\\ \langle\psi|\widetilde{P_{A_{0}}^{0}}^{r}\widetilde{P_{B_{0}}^{0}}^{r}|\psi\rangle=\frac{(1-\epsilon^{r})}{2}+\frac{\epsilon^{r}}{4}, \\ \langle\psi|\widetilde{P_{A_{1}}^{a}}^{r}|\psi\rangle=\frac{1}{2},\ \ for \ a\in\{0,1\},  \\ \langle\psi|\widetilde{P_{A_{1}}^{a}}^{r}\widetilde{W_a^{\alpha}}^{r}|\psi\rangle=\sqrt{2}(1-\epsilon^{r}),\ \ \ a\in\{0,1\},
\end{gather}
then there exist fixed CPTP maps $\Gamma_{A},\Gamma_{B},\Gamma_{C}$ as shown in Appendix A, such that 
\begin{equation}\nonumber
F((\Gamma_{A}\otimes\Gamma_{B}\otimes\Gamma_{C})(\tau_{ABC}^{r}),GHZ^{r}_{A^{\prime}B^{\prime}C^{\prime}})\geq G(\epsilon^{r}),
\end{equation}
for all $r\in\{0,1,\dots,7\}$. The function $G(x)$ is defined in main text as a function about lower bound of fidelity and white noise $\epsilon^{r}$. It is a numerical solution from SDP.
\end{lemma}
\begin{proof}
For $r=0$, we have given the detailed process of SDP to derive this result in Sec. \ref{sec:robust}. The CPTP maps are fixed for all $r\in\{0,1,\dots,7\}$. For different $r$, the observables in above correlations are all equivalent to the $r=0$ case, up to local unitary transformations. Thus, the lower bound of fidelity for different $r$ have the same form, i.e., they have a same function $G(x)$.
\end{proof}
Now, we start to prove the Theorem 2 that finding the lower bound on the quality of the unknown real measurement $\{R_{R_{1}R_{2}R_{3}}^{r}\}_{r=0}^{7}$. As $GHZ^{r}_{A^{\prime}B^{\prime}C^{\prime}}$ are pure states and from Eq.$~$(\ref{eq:abc}), one has
\begin{equation}\nonumber
\begin{split}
&\quad p_{r}F((\Gamma_{A}\otimes\Gamma_{B}\otimes\Gamma_{C})(\tau_{ABC}^{r}),GHZ^{r}_{A^{\prime}B^{\prime}C^{\prime}})\\&=p_{r}\langle (\Gamma_{A}\otimes\Gamma_{B}\otimes\Gamma_{C})(\tau_{ABC}^{r}),GHZ^{r}_{A^{\prime}B^{\prime}C^{\prime}}\rangle\\&=\langle \sigma_{A^{\prime}R_{1}}\otimes\sigma_{B^{\prime}R_{2}}\otimes\sigma_{C^{\prime}R_{3}},GHZ^{r}_{A^{\prime}B^{\prime}C^{\prime}}\otimes R_{R_{1}R_{2}R_{3}}^{r}\rangle.
\end{split}
\end{equation}
From the Lemma 3, there is
\begin{equation}
\langle \sigma_{A^{\prime}R_{1}}\otimes\sigma_{B^{\prime}R_{2}}\otimes\sigma_{C^{\prime}R_{3}},GHZ^{r}_{A^{\prime}B^{\prime}C^{\prime}}\otimes R_{R_{1}R_{2}R_{3}}^{r}\rangle \geq p_{r}G(\epsilon^{r}).
\label{eq:pr}
\end{equation}
To derive the main result, one should construct unital CPTP maps $\Lambda_{R_{1}}:\ \mathcal{L}(\mathcal{H}_{R_{1}})\to \mathcal{L}(\mathcal{H}_{A^{\prime}})$, $\Lambda_{R_{2}}:\ \mathcal{L}(\mathcal{H}_{R_{2}})\to \mathcal{L}(\mathcal{H}_{B^{\prime}})$ and $\Lambda_{R_{3}}:\ \mathcal{L}(\mathcal{H}_{R_{3}})\to \mathcal{L}(\mathcal{H}_{C^{\prime}})$, and then find the lower bound on $\langle \Lambda_{R_{1}}\otimes\Lambda_{R_{2}}\otimes\Lambda_{R_{3}}(R_{R_{1}R_{2}R_{3}}^{r}),GHZ_{A^{\prime}B^{\prime}C^{\prime}}^{r}\rangle$. Let $\lambda_{A^{\prime}R_{1}},\lambda_{B^{\prime}R_{2}}$ and $\lambda_{C^{\prime}R_{3}}$ be the Choi states of the maps $\Lambda_{R_{1}},\Lambda_{R_{2}}$ and $\Lambda_{R_{3}}$. One has
\begin{equation}
\begin{split}
&\quad\langle \Lambda_{R_{1}}\otimes\Lambda_{R_{2}}\otimes\Lambda_{R_{3}}(R_{R_{1}R_{2}R_{3}}^{r}),GHZ_{A^{\prime}B^{\prime}C^{\prime}}^{r}\rangle\\&=\langle Tr_{R_{1}R_{2}R_{3}}[(\lambda_{A^{\prime}R_{1}}\otimes\lambda_{B^{\prime}R_{2}}\otimes\lambda_{C^{\prime}R_{3}})(I_{A^{\prime}B^{\prime}C^{\prime}}\otimes(R_{R_{1}R_{2}R_{3}}^{r})^{T})]\\&\qquad\qquad,GHZ_{A^{\prime}B^{\prime}C^{\prime}}^{r}\rangle\\&=\langle \lambda_{A^{\prime}R_{1}}\otimes\lambda_{B^{\prime}R_{2}}\otimes\lambda_{C^{\prime}R_{3}},GHZ_{A^{\prime}B^{\prime}C^{\prime}}^{r}\otimes (R_{R_{1}R_{2}R_{3}}^{r})^{T}\rangle\\&=\langle \lambda_{A^{\prime}R_{1}}^{T}\otimes\lambda_{B^{\prime}R_{2}}^{T}\otimes\lambda_{C^{\prime}R_{3}}^{T},GHZ_{A^{\prime}B^{\prime}C^{\prime}}^{r}\otimes R_{R_{1}R_{2}R_{3}}^{r}\rangle.
\end{split}
\end{equation}
To utilize the relation in Eq.$~$(\ref{eq:pr}) into above equation, the Choi states should be constructed by $\sigma_{A^{\prime}R_{1}},\sigma_{B^{\prime}R_{2}},\sigma_{C^{\prime}R_{3}}$, respectively. One can bound the marginals $\sigma_{A^{\prime}},\sigma_{B^{\prime}}$ and $\sigma_{C^{\prime}}$ to guarantee the marginals of the constructed Choi states are proportional to $I$. From the Eq.$~$(\ref{eq:abc}), we have
\begin{equation}\nonumber
\begin{split}
F((\Gamma_{A}\otimes\Gamma_{B}\otimes\Gamma_{C})(\tau_{ABC}^{r}),GHZ^{r}_{A^{\prime}B^{\prime}C^{\prime}})&=F(\sigma_{A^{\prime}B^{\prime}C^{\prime}}^{r},GHZ^{r}_{A^{\prime}B^{\prime}C^{\prime}})\\&=\langle \sigma_{A^{\prime}B^{\prime}C^{\prime}}^{r},GHZ^{r}_{A^{\prime}B^{\prime}C^{\prime}}\rangle\\&\geq G(\epsilon^{r}). 
\end{split}
\end{equation}
Here, we adopt the definition in main text about notation $\widetilde{\{..\}}^{r}$ and define 
\begin{equation}\nonumber
\begin{split}
\sigma_{A^{\prime}B^{\prime}C^{\prime}}^{\prime}&=\sum\limits_{r=0}^{7}p_{r}(\widetilde{\sigma_{A^{\prime}B^{\prime}C^{\prime}}^{r}}^{r})^{\dag}\\&=\sum\limits_{r=0}^{7}p_{r}(U_{A^{\prime}}^{\prime r}\otimes U_{B^{\prime}}^{\prime r}\otimes U_{C^{\prime}}^{\prime r})\sigma_{A^{\prime}B^{\prime}C^{\prime}}^{r}(U_{A^{\prime}}^{\prime r}\otimes U_{B^{\prime}}^{\prime r}\otimes U_{C^{\prime}}^{\prime r})^{\dag}.
\end{split}
\end{equation}
By calculation, one has
\begin{equation}
\begin{split}
F(\sigma_{A^{\prime}B^{\prime}C^{\prime}}^{\prime},GHZ^{0}_{A^{\prime}B^{\prime}C^{\prime}})&=\langle \sigma_{A^{\prime}B^{\prime}C^{\prime}}^{\prime},GHZ^{0}_{A^{\prime}B^{\prime}C^{\prime}}\rangle\\&=\sum\limits_{r=0}^{7}p_{r}\langle \sigma_{A^{\prime}B^{\prime}C^{\prime}}^{r},GHZ^{r}_{A^{\prime}B^{\prime}C^{\prime}}\rangle\\&\geq \sum\limits_{r=0}^{7}p_{r}G(\epsilon^{r})=q.
\label{eq:fidelity}
\end{split}
\end{equation}
Furthermore, the spectrum of $\sigma_{A^{\prime}}$ is the same as $\sigma_{A^{\prime}}^{\prime}$ because of 
\begin{equation}\nonumber
\begin{split}
\sigma_{A^{\prime}}&=Tr_{R_{1}}\sigma_{A^{\prime}R_{1}}=Tr_{B^{\prime}C^{\prime}R_{1}R_{2}R_{3}}(\sigma_{A^{\prime}R_{1}}\otimes\sigma_{B^{\prime}R_{2}}\otimes\sigma_{C^{\prime}R_{3}})\\&=\Sigma_{r}Tr_{B^{\prime}C^{\prime}R_{1}R_{2}R_{3}}(R_{R_{1}R_{2}R_{3}}^{r}(\sigma_{A^{\prime}R_{1}}\otimes\sigma_{B^{\prime}R_{2}}\otimes\sigma_{C^{\prime}R_{3}}))\\&=\Sigma_{r}p_{r}Tr_{B^{\prime}C^{\prime}}\sigma_{A^{\prime}B^{\prime}C^{\prime}}^{r}=\Sigma_{r}p_{r}\sigma_{A^{\prime}}^{r}=\sigma_{A^{\prime}}^{\prime},
\end{split}
\end{equation}
where we use that $\Sigma_{r=0}^{2^{N}-1}R_{R_{1}R_{2}R_{3}}^{r}=I.$ 
Next, we will bound the spectrum of $\sigma_{A^{\prime}}^{\prime}$. One can always find a pure state $\sigma_{A^{\prime}B^{\prime}C^{\prime}}^{\prime}$ to achieve the upper and lower bounds. Without loss of generality, let $\sigma_{A^{\prime}B^{\prime}C^{\prime}}^{\prime}=\alpha|000\rangle+\beta|111\rangle$. By inequality Eq.$~$(\ref{eq:fidelity}), $0.5< q\leq 1$ and $\alpha^{2}+\beta^{2}=1$, one has $\frac{1-2\sqrt{q(1-q)}}{2}\leq\alpha^{2}\leq \frac{1+2\sqrt{q(1-q)}}{2}$. Thus, $spectrum(\sigma_{A^{\prime}})=spectrum(\sigma_{A^{\prime}}^{\prime})\in[\frac{1-2\sqrt{q(1-q)}}{2}, \frac{1+2\sqrt{q(1-q)}}{2}]$. One can write the spectrum of $\sigma_{A^{\prime}}$ as  
\begin{equation}\nonumber
spectrum(\sigma_{A^{\prime}})=\{\frac{1-\eta_{A^{\prime}}}{2},\frac{1+\eta_{A^{\prime}}}{2}\},
\end{equation}
where $0\leq\eta_{A^{\prime}}\leq 2\sqrt{q(1-q)}<1$. The same bound on $\eta_{B^{\prime}}$ and $\eta_{C^{\prime}}$ will be obtained in a similar way for 
\begin{equation}\nonumber
spectrum(\sigma_{B^{\prime}})=\{\frac{1-\eta_{B^{\prime}}}{2},\frac{1+\eta_{B^{\prime}}}{2}\}
\end{equation}
and
\begin{equation}\nonumber
spectrum(\sigma_{C^{\prime}})=\{\frac{1-\eta_{C^{\prime}}}{2},\frac{1+\eta_{C^{\prime}}}{2}\}.
\end{equation}

Now, the detailed form of Choi states are:
\begin{equation}
\begin{split}
&\lambda_{A^{\prime}R_{1}}^{T}=(\sigma_{A^{\prime}}^{-1/2}\otimes I)\sigma_{A^{\prime}R_{1}}(\sigma_{A^{\prime}}^{-1/2}\otimes I),\\&\lambda_{B^{\prime}R_{2}}^{T}=\frac{2}{1+\eta_{B^{\prime}}}\sigma_{B^{\prime}R_{2}}+\sigma_{R_{2}}\otimes(I-\frac{2}{1+\eta_{B^{\prime}}}\sigma_{B^{\prime}}),\\&\lambda_{C^{\prime}R_{3}}^{T}=\frac{2}{1+\eta_{C^{\prime}}}\sigma_{C^{\prime}R_{3}}+\sigma_{R_{3}}\otimes(I-\frac{2}{1+\eta_{C^{\prime}}}\sigma_{C^{\prime}}),
\end{split}
\end{equation}
where $\sigma_{A^{\prime}}=Tr_{R_{1}}\sigma_{A^{\prime}R_{1}}, \sigma_{B^{\prime}}=Tr_{R_{2}}\sigma_{B^{\prime}R_{2}},\sigma_{C^{\prime}}=Tr_{R_{3}}\sigma_{C^{\prime}R_{3}}$ and $\sigma_{R_{2}}=Tr_{B^{\prime}}\sigma_{B^{\prime}R_{2}},\sigma_{R_{3}}=Tr_{C^{\prime}}\sigma_{C^{\prime}R_{3}}$. 
As the $spectrum(\sigma_{B^{\prime}})=\{\frac{1-\eta_{B^{\prime}}}{2},\frac{1+\eta_{B^{\prime}}}{2}\}$ and $spectrum(\sigma_{C^{\prime}})=\{\frac{1-\eta_{C^{\prime}}}{2},\frac{1+\eta_{C^{\prime}}}{2}\}$ are bounded by $0\leq\eta_{B^{\prime}}\leq 2\sqrt{q(1-q)} , \ \ 0\leq\eta_{C^{\prime}}\leq 2\sqrt{q(1-q)}$, the $\sigma_{R_{3}}\otimes(I-\frac{2}{1+\eta_{C^{\prime}}}\sigma_{C^{\prime}}),\sigma_{R_{2}}\otimes(I-\frac{2}{1+\eta_{B^{\prime}}}\sigma_{B^{\prime}})$ are positive semidefinite. 
Thus, one has
\begin{equation}\nonumber
\lambda_{A^{\prime}R_{1}}^{T}\otimes\lambda_{B^{\prime}R_{2}}^{T}\otimes\lambda_{C^{\prime}R_{3}}^{T}\geq \lambda_{A^{\prime}R_{1}}^{T}\otimes\frac{2}{1+\eta_{B^{\prime}}}\sigma_{B^{\prime}R_{2}}\otimes \frac{2}{1+\eta_{C^{\prime}}}\sigma_{C^{\prime}R_{3}}.
\end{equation}
From the Lemma 3 in the supplement material for Ref.$~$\cite{renou2018self}, the inequality 
\begin{equation}
\begin{split}
\lambda_{A^{\prime}R_{1}}^{T}\geq s(\eta_{A^{\prime}})\sigma_{A^{\prime}R_{1}}-t(\eta_{A^{\prime}})\frac{I}{2}\otimes\sigma_{R_{1}}
\label{eq:choidayu}
\end{split}
\end{equation}
holds, where $s(x)=\frac{2}{\sqrt{1-x^{2}}},t(x)=\frac{4}{\sqrt{1-x^{2}}}-\frac{4}{1+x}$ and $\sigma_{R_{1}}=Tr_{A^{\prime}}\sigma_{A^{\prime}R_{1}}$. Therefore, one has
\begin{equation}\nonumber
\begin{split}
\lambda_{A^{\prime}R_{1}}^{T}
&\otimes\lambda_{B^{\prime}R_{2}}^{T}\otimes\lambda_{C^{\prime}R_{3}}^{T}\geq (s(\eta_{A^{\prime}})\sigma_{A^{\prime}R_{1}}-t(\eta_{A^{\prime}})\frac{I}{2}\otimes\sigma_{R_{1}})\\&\otimes \frac{2}{1+\eta_{B^{\prime}}}\sigma_{B^{\prime}R_{2}} \otimes \frac{2}{1+\eta_{C^{\prime}}}\sigma_{C^{\prime}R_{3}},
\end{split}
\end{equation}
where the inequality is from Eq.$~$(\ref{eq:choidayu}) and positive semidefinite matrices $\frac{2}{1+\eta_{B^{\prime}}}\sigma_{B^{\prime}R_{2}} , \ \ \frac{2}{1+\eta_{C^{\prime}}}\sigma_{C^{\prime}R_{3}}$. As 
\begin{equation}\nonumber
\begin{split}
&\langle \sigma_{R_{1}}\otimes\sigma_{R_{2}}\otimes \sigma_{R_{3}},R_{R_{1}R_{2}R_{3}}^{r}\rangle\\&=Tr_{R_{1}R_{2}R_{3}}((\sigma_{R_{1}}\otimes\sigma_{R_{2}}\otimes \sigma_{R_{3}})\cdot R_{R_{1}R_{2}R_{3}}^{r})\\&=Tr_{A^{\prime}B^{\prime}C^{\prime}R_{1}R_{2}R_{3}}((\sigma_{A^{\prime}R_{1}}\otimes\sigma_{B^{\prime}R_{2}}\otimes \sigma_{C^{\prime}R_{3}})\cdot R_{R_{1}R_{2}R_{3}}^{r})\\&=p_{r}Tr_{A^{\prime}B^{\prime}C^{\prime}}\sigma_{A^{\prime}B^{\prime}C^{\prime}}^{r}=p_{r},
\end{split}
\end{equation}
one has $\quad\langle I\otimes\sigma_{R_{1}}\otimes I\otimes\sigma_{R_{2}}\otimes \sigma_{C^{\prime}R_{3}},GHZ_{A^{\prime}B^{\prime}C^{\prime}}^{r}\otimes R_{R_{1}R_{2}R_{3}}^{r}\rangle=\frac{1}{2}\langle \sigma_{R_{1}}\otimes\sigma_{R_{2}}\otimes \sigma_{R_{3}},R_{R_{1}R_{2}R_{3}}^{r}\rangle=\frac{p_{r}}{2}.$ Then, one arrives at
\begin{equation}\nonumber
\begin{split}
&\quad\langle \lambda_{A^{\prime}R_{1}}^{T}\otimes\lambda_{B^{\prime}R_{2}}^{T}\otimes\lambda_{C^{\prime}R_{3}}^{T},GHZ_{A^{\prime}B^{\prime}C^{\prime}}^{r}\otimes R_{R_{1}R_{2}R_{3}}^{r}\rangle\\&\geq \frac{4s(\eta_{A^{\prime}})p_{r}G(\epsilon^{r})-t(\eta_{A^{\prime}})p_{r}}{(1+\eta_{B^{\prime}})(1+\eta_{C^{\prime}})}.
\end{split}
\end{equation}
The inequality is derived from the fact that the fidelity can only increase after tracing out the subsystem. Now, we can obtain 
\begin{equation}\nonumber
\begin{split}
&\mathcal{Q}(\mathcal{R},\mathcal{P})\geq \frac{1}{8}\sum\limits_{r=0}^{7}\langle \lambda_{A^{\prime}R_{1}}^{T}\otimes\lambda_{B^{\prime}R_{2}}^{T}\otimes\lambda_{C^{\prime}R_{3}}^{T},GHZ_{A^{\prime}B^{\prime}C^{\prime}}^{r}\otimes R_{R_{1}R_{2}R_{3}}^{r}\rangle\\&\geq \frac{1}{8}(\frac{4s(\eta_{A^{\prime}})\sum_{r=0}^{7}p_{r}G(\epsilon^{r})-t(\eta_{A^{\prime}})\sum_{r=0}^{7}p_{r}}{(1+\eta_{B^{\prime}})(1+\eta_{C^{\prime}})})
\\&=\frac{4s(\eta_{A^{\prime}})q-t(\eta_{A^{\prime}})}{8(1+\eta_{B^{\prime}})(1+\eta_{C^{\prime}})}.
\end{split}
\end{equation}
As $0.5< q\leq 1$, the numerator is positive. Hence, one obtains the result
\begin{equation}\nonumber
\begin{split}
\mathcal{Q}(\mathcal{R},\mathcal{P})&\geq\frac{1}{2(1+2\sqrt{q(1-q)})^{2}}(\frac{2q-1}{\sqrt{(1-\eta_{A^{\prime}}^{2})}}+\frac{1}{(1+\eta_{A^{\prime}})})\\&\geq\frac{1}{2(1+2\sqrt{q(1-q)})^{2}}\\&\cdot\mathop{min}\limits_{u\in[0,2\sqrt{q(1-q)}]}(\frac{2q-1}{\sqrt{(1-u^{2})}}+\frac{1}{(1+u)}).
\end{split}
\end{equation}
Here, we have presented a lower bound for the quality of unknown joint measurement performed by Roy under certain white noise. The quality implies the ability that the unknown measurement try to simulate the ideal three-qubit GHZ-state measurement. Therefore, a robust self-testing statement for three-qubit GHZ-state measurement has been shown.

\bibliographystyle{unsrt}

\end{document}